\documentclass[lettersize,journal]{IEEEtran}
\usepackage{amsmath,amsfonts}
\usepackage{algorithmic}
\usepackage{algorithm}
\usepackage{array}
\usepackage[caption=false,font=normalsize,labelfont=sf,textfont=sf]{subfig}
\usepackage{textcomp}
\usepackage{stfloats}
\usepackage{url}
\usepackage{verbatim}
\usepackage{graphicx}
\usepackage{cite}

\usepackage{color}
\usepackage{amsthm}

\newtheorem{theorem}{Theorem}

\theoremstyle{definition}

\DeclareFontFamily{U}{wncy}{}
\DeclareFontShape{U}{wncy}{m}{n}{<->wncyr10}{}
\DeclareSymbolFont{mcy}{U}{wncy}{m}{n}
\DeclareMathSymbol{\Sh}{\mathord}{mcy}{"58}

\begin{document}

\title{Practical Modulo Sampling: Mitigating High-Frequency Components}

\author{Yhonatan Kvich, ~\IEEEmembership{Graduate Student Member,~IEEE}, Shlomi Savariego, Moshe Namer, Yonina C. Eldar,~\IEEEmembership{Fellow,~IEEE}
\thanks{Y. Kvich, S. Savariego, M. Namer and Y. C. Eldar are with the Faculty of Mathematics and Computer Science, Weizmann Institute of Science, Israel.}
\thanks{This research was supported by the European Research Council (ERC) under the European Union’s Horizon 2020 research and innovation program (grant No. 101000967), by the Israel Science Foundation (grant No. 536/22), by the Tom and Mary Beck Center for Renewable Energy as part of the Institute for Environmental Sustainability (IES) at the Weizmann Institute of Science.}
}
        
\maketitle

\begin{abstract}
	Recovering signals within limited dynamic range (DR) constraints remains a central challenge for analog-to-digital converters (ADCs). To prevent data loss, an ADC’s DR typically must exceed that of the input signal. Modulo sampling has recently gained attention as a promising approach for addressing DR limitations across various signal classes. However, existing methods often rely on ideal ADCs capable of capturing the high frequencies introduced by the modulo operator, which is impractical in real-world hardware applications. This paper introduces an innovative hardware-based sampling approach that addresses these high-frequency components using an analog mixer followed by a Low-Pass Filter (LPF). This allows the use of realistic ADCs, which do not need to handle frequencies beyond the intended sampling rate. Our method eliminates the requirement for high-specification ADCs and demonstrates that the resulting samples are equivalent to those from an ideal high-spec ADC. Consequently, any existing modulo recovery algorithm can be applied effectively. We present a practical hardware prototype of this approach, validated through both simulations and hardware recovery experiments. Using a recovery method designed to handle quantization noise, we show that our approach effectively manages high-frequency artifacts, enabling reliable modulo recovery with realistic ADCs. These findings confirm that our hardware solution not only outperforms conventional methods in high-precision settings but also demonstrates significant real-world applicability.

\end{abstract}


\section{Introduction}
\label{sec:intro}

Analog-to-digital converters (ADCs) play a vital role in converting analog signals into a digital format for processing within digital signal processing systems. The cost and power requirements of ADCs escalate with higher sampling rates, emphasizing the importance of operating at the minimum necessary rate for efficient sampling \cite{eldar2015sampling,mishali2011sub}. The Shannon-Nyquist sampling theorem states that bandlimited (BL) signals can be accurately represented by uniform samples taken at a rate at least double the maximum frequency present in the signal. Sampling close to the signal's Nyquist rate is advantageous.
Another critical consideration is the dynamic range (DR) of ADCs. To prevent signal clipping and consequent information loss, as depicted in Fig. \ref{fig:compare}(a), an ADC's DR must exceed that of the input analog signal.

\begin{figure}[htb]
	\begin{minipage}[b]{\linewidth}
		\centering
		\centerline{\includegraphics[width=\columnwidth]{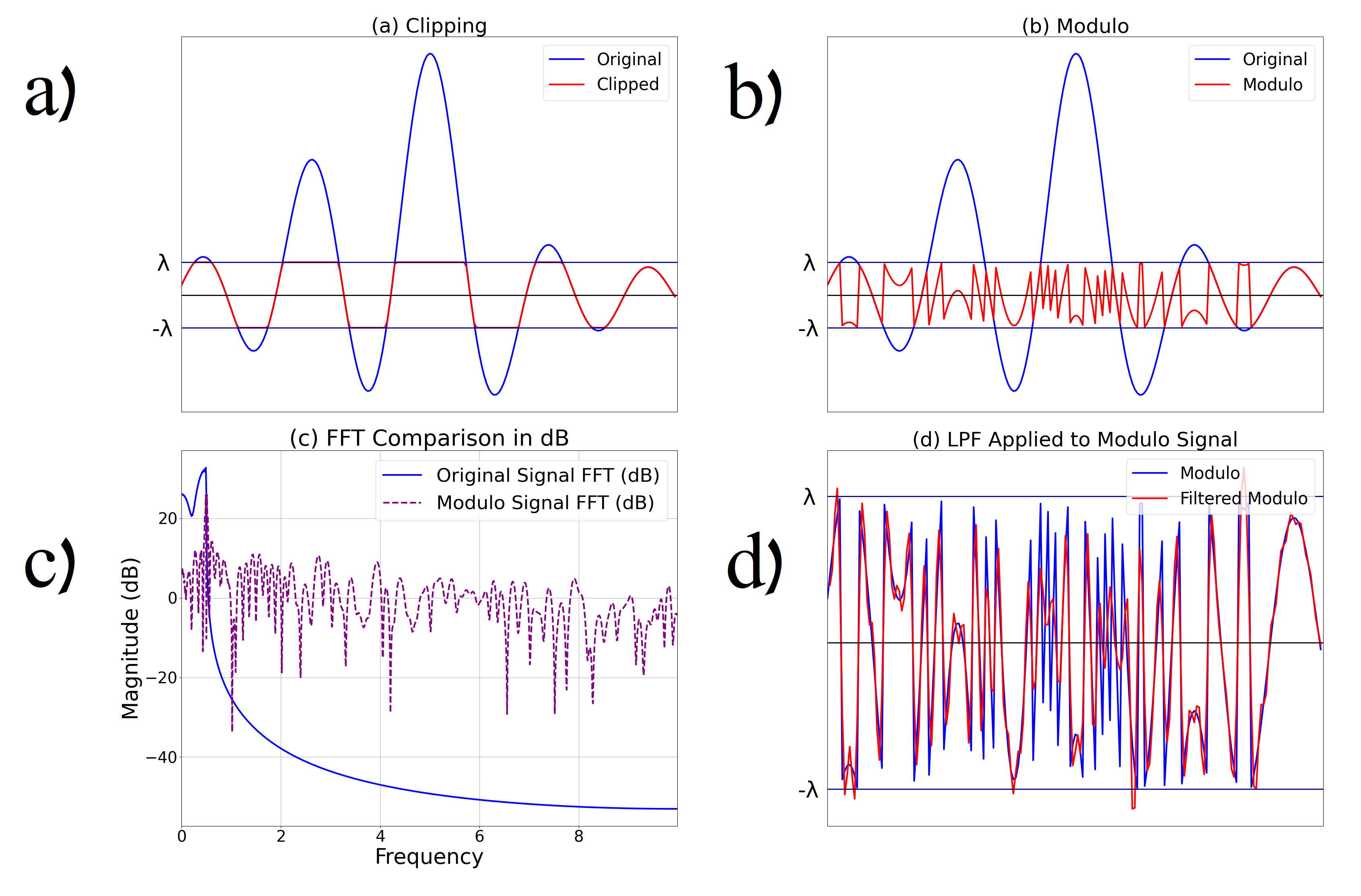}}
	\end{minipage}
	\caption{(a) BL signal alongside its clipped version, sampled with an ADC with dynamic range \([- \lambda, \lambda]\). (b) BL signal and the result of applying the modulo operation to it. (c) Fourier transforms of both the original BL signal (blue) and the modulo signal (red), highlighting the frequency components introduced by the modulo operation. (d) Modulo signal (blue) and the filtered modulo signal after removing frequencies above 5 (red), illustrating the effect of bandwidth limitation.}
	\label{fig:compare}
\end{figure}

Various strategies have emerged to address clipping. Techniques without preprocessing leverage inherent correlation in BL signal samples taken above the Nyquist rate, reconstructing missing information due to clipping through oversampling \cite{marks1983restoring, marks1984error}. Alternatively, exploiting spectral gaps in multiband systems can discern original from clipped signals \cite{abel1991restoring, rietman2008clip}. However, these methods either demand significant oversampling \cite{marks1983restoring, marks1984error} or prior knowledge of spectral gaps \cite{abel1991restoring, rietman2008clip} without providing theoretical guarantees. Clipping can also be mitigated through attenuation, though this risks reducing low-amplitude signals below the noise floor. Variable gain attenuators, like automatic gain controls (AGCs) and companders, adjust to preserve signal integrity without disproportionately affecting smaller amplitude signals. AGCs use feedback-regulated amplifiers to maintain consistent output levels \cite{perez2011automatic, mercy1981review}, while companding adjusts gain inversely proportional to signal amplitude. However, companding, like clipping, introduces nonlinear distortion and broadens signal bandwidth.
Implementing companding-based solutions at minimal sampling rates faces challenges due to monotonicity, differentiability requirements, and finite energy output \cite{landau1960recovery, landau1961recovery}.

An alternative strategy involves applying a modulo operation to the input signal before sampling to limit its DR, a technique applied in high-dynamic-range ADCs, or self-reset ADCs, in imaging contexts \cite{park2007wide, sasagawa2015implantable, yuan2009activity, krishna2019unlimited}. This approach, alongside storing modulo signal samples, often involves capturing additional data, such as folding extent for each sample or folding direction, complicating sampling circuitry while simplifying signal reconstruction from folded samples. See Fig. \ref{fig:compare}(b) for a visual representation.

The concept of unlimited sampling, relies solely on folded or modulo samples for signal recovery, demonstrating that sampling above the Nyquist rate enables unique identification of BL signals from modulo samples \cite{bhandari2020unlimited}. This method, extending Itoh’s unwrapping algorithm, shows that a sufficient oversampling rate allows computation of original signal higher-order differences from modulo samples, facilitating signal reconstruction through cumulative summation of these differences. However, this technique's effectiveness diminishes in noisy environments, requiring significantly higher oversampling rates for reliable recovery \cite{bhandari2020unlimited}. Subsequent improvements have aimed to reduce required sampling rates, apply the technique to various signal models, and explore hardware implementations for high-dynamic-range ADCs using modulo operations \cite{romanov2019above, lu2020high, bhandari2021unlimited, rudresh2018wavelet, bhandari2018unlimited, musa2018generalized, prasanna2020identifiability, ji2019folded, Bhandari_Krahmer_2020, Guo_Bhandari_2023, zhao2023unlimited, zhang2024line}. Despite these advancements, challenges remain regarding missing theoretical guarantees, stability concerns, the need for smooth and monotone operators, and reliance on higher-than-Nyquist sampling rates.

Azar et al. \cite{azar2022robust, Azar_Mulleti_Eldar_2022a} introduced a recovery algorithm capable of reconstructing BL signals from modulo samples slightly above the Nyquist rate, even under various noise conditions. This approach was further refined in \cite{Shah_Mulleti_Eldar_2023} through the incorporation of a sparsity assumption to reduce computational complexity. Additionally, Mulleti et al. \cite{mulleti2024modulo} explored applying modulo sampling to finite-rate-of-innovation (FRI) signals. The methodologies developed for both BL and FRI signals have been successfully implemented in hardware \cite{mulleti2023hardware}.
Bernardo et al. \cite{bernardo2024modulo} introduced a new recovery algorithm that utilizes 1-bit cross-level information. The authors provided proof of recovery and demonstrated robustness to quantization noise for oversampling rates greater than 3, assuming the system contains at least 4 bits. They showed superior results compared to classical sampling at high sampling rates.
While the bandlimitation assumption often serves as a practical approximation, numerous signals offer more precise representations through alternative bases or possess distinct structures within the Fourier domain \cite{unser2000sampling, eldar2009compressed, bhandari2018unlimited, rudresh2018wavelet, prasanna2020identifiability, ji2022unlimited, mulleti2024modulo}. Shift-invariant (SI) spaces are especially important in sampling theory, where signals in these spaces are represented as linear combinations of shifts of a set of generating functions \cite{eldar2015sampling, eldar2009compressed, deboor1994structure, christensen2004oblique, aldroubi2001nonuniform, bhandari2011shift}. Kvich and Eldar \cite{kvich2024Modulo, kvich2024Modulo2} proposed a recovery method for SI signals using modulo sampling, requiring a sampling rate only slightly above the minimal threshold.



An ADC fundamentally operates by alternating between two phases: track-and-hold (T/H) and quantization \cite[Section 14.3.4]{eldar2015sampling}. During the T/H phase, the ADC follows the signal's variations. Once an accurate tracking is achieved, the ADC holds this value steady, allowing the quantizer to transform the signal's amplitude into a digital representation. These steps must be completed before acquiring the next sample. It is common in signal processing to idealize the ADC as a pointwise sampler that captures the signal at a consistent rate of samples per second. Nonetheless, due to the inherent limitations of analog circuits, the T/H function has a finite frequency tracking capability and cannot follow signals that change too quickly. Practically, a Low-Pass Filter (LPF) with a specific cutoff frequency approximates the T/H function's bandwidth limitation. Commercial ADCs typically specify this internal LPF cutoff frequency to be higher than the maximum sampling rate, but within the same order of magnitude \cite{mishali2011sub}.

This paper presents a hardware-based approach to addressing the high-frequency components introduced by the analog modulo operator, which results in a signal with an expanded bandwidth. Fig. \ref{fig:compare}(c) illustrates this, showing that the ``folded" signal contains significantly higher frequencies than the original signal. Most existing works assume an ideal pointwise sampler, where the ADC can capture this entire extended band, necessitating a much more advanced ADC than required for the original signal. However, this assumption presents challenges in practical scenarios, as high-specification ADCs are neither cost-effective nor feasible. Directly using an ADC suited to the original signal’s bandwidth on the ``folded" signal would remove these high-frequency components during sampling, introducing distortions that current recovery methods do not address. Fig. \ref{fig:compare}(d) demonstrates this effect, as the modulo signal becomes distorted after passing through a LPF, even when the cutoff frequency is well above the Nyquist rate of the input signal.

We propose a new sampling approach inspired by the modulated wideband converter (MWC) \cite{eldar2015sampling, mishali2010theory, mishali2011xampling}. Our system uses an analog mixer applied to the ``folded" signal, followed by a LPF and sampling. By addressing the high-frequency components in the analog domain, this approach eliminates the need for ADCs capable of handling the entire extended bandwidth. ADCs suited to the input signal’s original band rather than the wide-band ``folded" signal will be referred to as ``realistic ADCs". We provide proof that our approach yields samples identical to those obtained from a high-specification ADC. This allows the use of any existing modulo recovery algorithm for BL signals, enabling signal recovery at any rate slightly above the Nyquist rate of the original signal \cite{azar2022robust, Azar_Mulleti_Eldar_2022a, mulleti2023hardware, Shah_Mulleti_Eldar_2023}. 
We then present a hardware prototype based on this theoretical approach, along with comprehensive simulations and hardware recovery experiments using only realistic ADCs. Our experiments employ a modulo recovery algorithm designed to handle quantization noise \cite{bernardo2024modulo}.

We perform several comparisons to evaluate the effectiveness of our approach. First, we consider ideal modulo sampling, where a pointwise sampler is used to capture the wideband folded signal without any limitations on the ADC's bandwidth capabilities. Second, we compare against directly sampling the ``folded" signal using a realistic ADC, where a LPF is applied prior to sampling to handle bandwidth limitations. Finally, we evaluate performance against a classical ADC with no DR limitations but constrained to a finite number of bits. While this configuration provides a useful benchmark, it is not very realistic due to the significant power requirements needed to overcome DR constraints.
We demonstrate that directly using a realistic ADC to sample the folded signal leads to errors that are orders of magnitude larger than quantization noise, resulting in poor reconstruction accuracy. In contrast, our proposed approach effectively mitigates the errors caused by ignoring high-frequency components, achieving performance that is close to the ideal sampler case. The classical sampler performs worse than both the ideal modulo sampler and our proposed realistic modulo sampling method. These comparisons underscore the importance and benefits of addressing high-frequency artifacts using realistic hardware solutions.


This paper is organized as follows. Section \ref{sec:background} provides background information and the problem statement, highlighting the challenges posed by high-frequency components in modulo sampling. We then introduce our sampling approach, along with a proof that the resulting samples align with the ideal modulo samples. Additionally, we present the modulo recovery method used in this context, specifically designed to address quantization noise.
Section \ref{sec:hardware} describes the design and implementation of our realistic modulo hardware prototype.
Section \ref{sec:results} presents the results from both simulations and hardware experiments, comparing our proposed method against ideal modulo sampling, direct sampling with realistic ADCs, and classical ADCs with no DR limitation.
Finally, Section \ref{sec:conclusion} offers conclusions and discusses the implications of our findings for future research and practical applications.

Throughout this paper, the following notations are used: The $\ell^2$ norm is denoted as \(\|\cdot\|_2\), and the infinity norm is written as \(\|\cdot\|_\infty\). For a given sequence \(a[n]\) with finite norm, its Discrete-Time Fourier Transform (DTFT) is \(A(e^{j\omega}) := \sum_{n\in\mathbb{Z}} a[n] e^{-j\omega n}\). Similarly, for a function \(x(t)\) with finite norm, its Continuous-Time Fourier Transform (CTFT) is defined as \(X(\omega) := \int_{t\in\mathbb{R}} x(t) e^{-j\omega t} \), we also employ the notation \(\mathcal{F}\{x\}(\omega)\). The inverse operators are IDTFT and ICTFT, respectively. Convolution is denoted by \(\ast\).
\section{Background and System Design}
\label{sec:background}

\subsection{Preliminaries and Problem Statement}

Capturing signals without exceeding an ADC’s DR remains a challenge in signal processing. When DR limits are surpassed, clipping causes data loss. Expanding DR can help but adds quantization noise and requires power-hungry, high-resolution ADCs. A more efficient solution applies a nonlinear operation, like the modulo function, to compress, or “fold,” the signal within a set range. Following studies in \cite{azar2022robust, Azar_Mulleti_Eldar_2022a, mulleti2023hardware}, the modulo operation remaps real values into the interval \([- \lambda, \lambda]\) for any \(\lambda > 0\) as:
\begin{equation}\label{eq:mod}
	\mathcal{M}_\lambda x := ((x+\lambda) \mod 2\lambda) - \lambda .
\end{equation}

Azar et al. \cite{azar2022robust, Azar_Mulleti_Eldar_2022a} demonstrated the potential to ``unfold" a BL signal's samples by sampling just above the Nyquist frequency. They further introduced a signal recovery technique for such situations, named $B^2 R^2$.
Fig. \ref{fig:azar} depicts the outlined strategy for the recovery of modulo-transformed BL signals. An input BL signal undergoes LPF and is then treated with the analog modulo operation.
Subsequently, the signal is sampled at a rate $T_s < T$, where $T$ represents the Nyquist rate, defined as double the maximum frequency present in the bandwidth of $x(t)$.
Additionally, several strategies for the restoration of BL signals have been proposed, like the approach in \cite{Shah_Mulleti_Eldar_2023}, which leverages inherent sparsity properties and Iterative Shrinkage-Thresholding Algorithm (ISTA) to refine the method in \cite{azar2022robust, Azar_Mulleti_Eldar_2022a}. Other methodologies, as discussed in \cite{bhandari2020unlimited,romanov2019above} and \cite{Guo_Bhandari_2023}, introduced signal reconstruction techniques using higher-order differences and iterative signal sieving, respectively. Nevertheless, these methods necessitate sampling rates above Nyquist and exhibit less noise resilience.

\begin{figure}[htb]
	\begin{minipage}[b]{\linewidth}
		\centering
		\includegraphics[width=0.9\columnwidth]{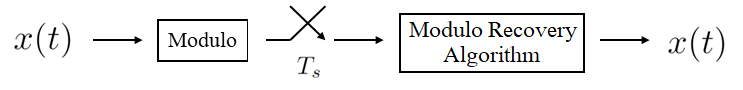}
	\end{minipage}
	\caption{Illustration of the modulo recovery framework for BL signals as discussed in \cite{azar2022robust, Azar_Mulleti_Eldar_2022a, mulleti2023hardware, Shah_Mulleti_Eldar_2023}.}
	\label{fig:azar}
\end{figure}

When applying the modulo operator to a BL signal \(x(t)\), the resulting ``folded" signal \(\mathcal{M}_\lambda x(t)\) exhibits a significantly broadened bandwidth. The ADC employed must therefore be capable of handling these increased frequency components.
Typically, the ADC used is designed for bandwidths significantly higher than that of the input signal \cite{mulleti2023hardware}, leading to samples that closely resemble \((\mathcal{M}_\lambda x)[nT_s]\).

In scenarios where the sampler is realistically capable of only handling the sampling rate \(T_s\), the ``folded" signal must be processed through an LPF before sampling, resulting in \(\text{LPF}(\mathcal{M}_\lambda x)[nT_s]\). This energy loss during filtering can significantly alter measurements, potentially undermining existing modulo recovery techniques that fail to consider such distortions. Our simulations indicate that this error may be orders of magnitude greater than the quantization error typically associated with classical sampling schemes. Therefore, even with perfect modulo recovery, disregarding the high-frequency components leads to recovery outcomes that are inferior to both modulo and classical sampling approaches.
Our prototype aims to overcome the challenges associated with sampling those high-frequency components using only realistic samplers.

\subsection{Sampling Approach}\label{sec:sampling}

Given a BL signal \(x(t)\) with a Nyquist rate of \(T\), our objective is to reconstruct the signal from modulo samples using samplers that operate within the range of \([-\lambda, \lambda]\) and incorporate a realistic sampler. This means that the sampler includes an internal LPF prior to sampling, with a cutoff frequency aligned to the sampling rate without significantly surpassing it.
Our methodology involves adding an analog multiplier with a $T_s$-periodic function $p(t)$, where $T_s < T$, prior to sampling. This will be followed by a LPF to turn the signal into a BL signal with Nyquist rate of $T_s$, resulting in \(y(t) = \text{LPF}(p(t)\mathcal{M}_\lambda x)\). The signal is than sampled at rate $T_s$, resulting in $y[n T_s] =  \text{LPF}(p(t)\mathcal{M}_\lambda x) [n T_s]$. The block diagram is shown in Fig. \ref{fig:sys comb}, and is inspired by the MWC where a similar approach was used for realistic Sub-Nyquist sampling \cite{mishali2010theory, mishali2011xampling}.

The signal \(y(t)\) is BL with a Nyquist rate of \(T_s\), which matches the sampling rate, ensuring that the ADC used only processes frequencies within this range. The following theorem demonstrates that selecting \(p(t)\) as a delta comb in time results in samples that are identical to the ideal pointwise samples of the ``folded" signal.

\begin{figure}[htb]
	\begin{minipage}[b]{\linewidth}
		\centering
		\centerline{\includegraphics[width=\columnwidth]{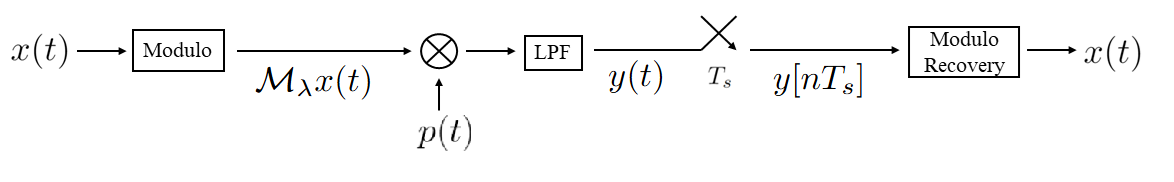}}
	\end{minipage}
	\caption{The proposed block diagram for sampling and modulo recovery using only realistic components, based on the diagram in Fig. \ref{fig:azar} and inspired by the MWC. The signal $x(t)$ it BL with Nyquist rate of $T$, the cutoff frequency of the LPF is $\frac{\pi}{T_s}$ and the sampling rate is $T_s < T$.}
	\label{fig:sys comb}
\end{figure}

\begin{theorem}\label{thm:mixer2}
    Consider the sampling process in Fig. \ref{fig:sys comb}. If $p(t) = \sum_{n\in\mathbb{Z}}{\delta(t-nT_s)}$, than $y[nT_s] = \mathcal{M}_\lambda x[nT_s]$, where $\mathcal{M}_\lambda$ is defined in (\ref{eq:mod}).
\end{theorem}

\begin{proof}
From the assumption on $p(t)$, we know that $P(\omega) = \sum_{l\in\mathbb{Z}}{\delta\Big(\omega - \frac{2\pi l}{T_s}\Big)}$. Denote $z(t) = p(t) \mathcal{M}_\lambda x(t)$, using CTFT and DTFT \cite{eldar2015sampling} we have 
	\begin{equation}\label{eq:thm1_1}
			Z (\omega) = \sum_{l\in\mathbb{Z}}{ \mathcal{F}\{\mathcal{M}_\lambda x\}\Big(\omega - \frac{2\pi l}{T_s}\Big)}.
	\end{equation}
Denote  $a[n] = \mathcal{M}_\lambda x[nT_s]$ we can from Poisson's formula that

\begin{equation}\label{eq:thm1_2}
			A\left(e^{j\omega T_s}\right) = \sum_{l\in\mathbb{Z}}{ \mathcal{F}\{\mathcal{M}_\lambda x\}\Big(\omega - \frac{2\pi l}{T_s}\Big)}.
	\end{equation}
Combining (\ref{eq:thm1_1}) and (\ref{eq:thm1_2}) yields
\begin{equation}\label{eq:thm1_3}
			A\left(e^{j\omega T_s}\right) = Z(\omega).
	\end{equation}
Since $y(t) = \text{LPF}(z(t))$ with cutoff frequency of $\frac{\pi}{T_s}$, we get that

\begin{equation}\label{eq:thm1_4}
		Y(\omega) = Z (\omega), \quad |\omega|\le\frac{\pi}{T_s}.
	\end{equation}
We sample $y(t)$ at rate of $T_s$, which is its Nyquist rate, thus 

\begin{equation}\label{eq:thm1_5}
		\text{DTFT}\left\{ y[nT_s] \right\}\left(e^{j\omega T_s}\right) = Y (\omega), \quad |\omega|\le\frac{\pi}{T_s}.
\end{equation}
From (\ref{eq:thm1_4}) and (\ref{eq:thm1_5}), we see that $a[n]$ and $y[nT_s]$ have the same DTFT therefore they are  identical.
\end{proof}


The theorem above provide an alternative approach to measuring \(\mathcal{M}_\lambda x[nT_s]\) using an ADC that does not need to handle higher frequencies. Once these samples are obtained, any existing modulo recovery method can be applied to deduce \(x[nT_s]\) and, consequently, the input signal \(x(t)\).



\subsection{Modulo Recovery and Quantization}\label{sec:mod_rec}

According to Theorem \ref{thm:mixer2}, the samples obtained through our approach match the ideal samples of the folded signal, allowing us to apply any existing modulo recovery method for BL signals. For our system, we adopt the method proposed by Bernardo et al. \cite{bernardo2024modulo}, which addresses quantization noise. This approach utilizes extra-bit information for signal reconstruction. Their system generates an additional 1-bit signal that indicates whether a fold has occurred since the last measurement, though it does not provide the number or direction of folds. Building on this concept, they developed a recovery method that guarantees accurate modulo recovery even in the presence of quantization noise, assuming an oversampling rate greater than 3 and at least 4 bits of resolution. To apply this method, we incorporated the extra-bit signal in our system, as depicted in Fig. \ref{fig:sys comb extrabit}.

\begin{figure}[htb]
	\begin{minipage}[b]{\linewidth}
		\centering
		\centerline{\includegraphics[width=\columnwidth]{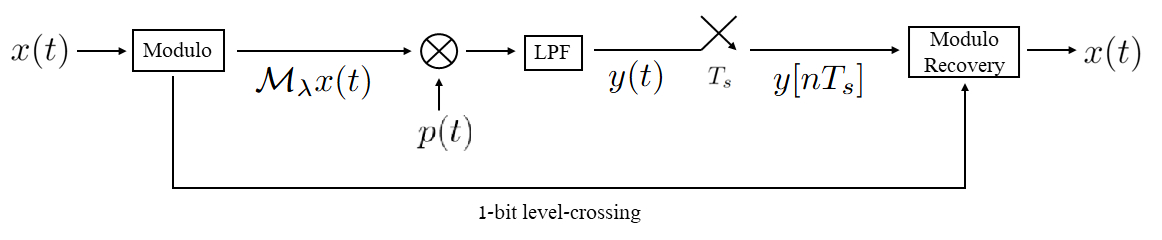}}
	\end{minipage}
	\caption{Block diagram incorporating extra-bit information for modulo recovery, adapted from Fig. \ref{fig:sys comb}, illustrating the integration of a 1-bit folding indicator with the existing system components to insure under quantization noise.}
	\label{fig:sys comb extrabit}
\end{figure}

In classical sampling the Mean Squared Error (MSE) quantization error is given as
\begin{equation}
	\label{eq:claasical_quant}
	E_{\text{classical}} = \frac{1}{\text{OF}} \frac{1}{12} \Delta^2_{\text{classical}}
\end{equation}
where $\Delta_{\text{classical}}$ is the quantization step and \(\text{OF}\) is the oversampling rate. Note that in this case, the DR encompasses the entire signal range, and all $b$ bits are utilized for quantization, meaning
\(
	\Delta_{\text{classical}} = \frac{1}{2^b - 1}  \|x\|_\infty 
\).
For modulo sampling we sample only the DR of $[-\lambda,\lambda]$, meaning the quantization error for modulo sampling is

\begin{equation}
	\label{eq:mod_quant}
	E_{\text{mod-Q}} = \frac{1}{\text{OF}} \frac{1}{12} \Delta^2_{\text{mod}}
\end{equation}
where
\(
	\Delta_{\text{mod}} = \frac{1}{2^b - 1} \lambda.
\)
Setting $\lambda = \frac{\|x\|_\infty}{\text{OF}-2}$ as in \cite{bernardo2024modulo}, yields the modulo quantization error

\begin{equation}
	E_{\text{mod-Q}} =\frac{1}{\text{OF}} \frac{1}{12(2^b - 1)^2} \Big( \frac{\|x\|_\infty}{\text{OF} - 2}\Big)^2.
\end{equation}
The MSE of the modulo sampling is \(\mathcal{O}(\text{OF}^{-3})\), which is a significant improvement over the classical MSE that scales as \(\mathcal{O}(\text{OF}^{-1})\).

Theorem \ref{thm:mixer2} assumes that $p(t)$ is a delta comb. In practice the number of combs is finite and we will assume it consists of $N=2000$ deltas in frequency, meaning \(
	p(t) = \sum_{k=-N}^{N}{e^{j \frac{k}{T_s} t} }\). This lead to an error between the measurements $y[n T_s]$ and the desired values $\mathcal{M}_\lambda x[n T_s]$. We denote the MSE of the difference as $E_{\text{mod-HF}}$. Because quantization occurs after the analog processing step, the errors are independent, resulting in an overall MSE for the modulo approach given by

\begin{equation}
\label{eq:mod_mse}
	E_{\text{mod}} = E_{\text{mod-HF}} + E_{\text{mod-Q}}.
\end{equation}

In the next section, we present the hardware prototype that implements our approach, followed by the results from both simulations and the hardware in Section \ref{sec:results}.

\section{Realistic Modulo Hardware Prototype}
\label{sec:hardware}

In this section, we provide details on the system architecture of the hardware prototype. The block diagram is shown in Fig. \ref{fig:sys comb extrabit}, and the hardware setup is illustrated in Fig. \ref{fig:HW_board}. In addition to our proposed sampling system, the hardware includes three additional channels: a classical sampler with infinite DR, a modulo sampler with an ideal pointwise sampler capable of handling much higher frequencies as seen in Fig. \ref{fig:azar}, and a modulo sampler where the folded signal passes through a LPF before sampling, simulating the direct use of a practical ADC without our approach.

\begin{figure*}[ht]
  \centering
  \includegraphics[width=0.8\textwidth]{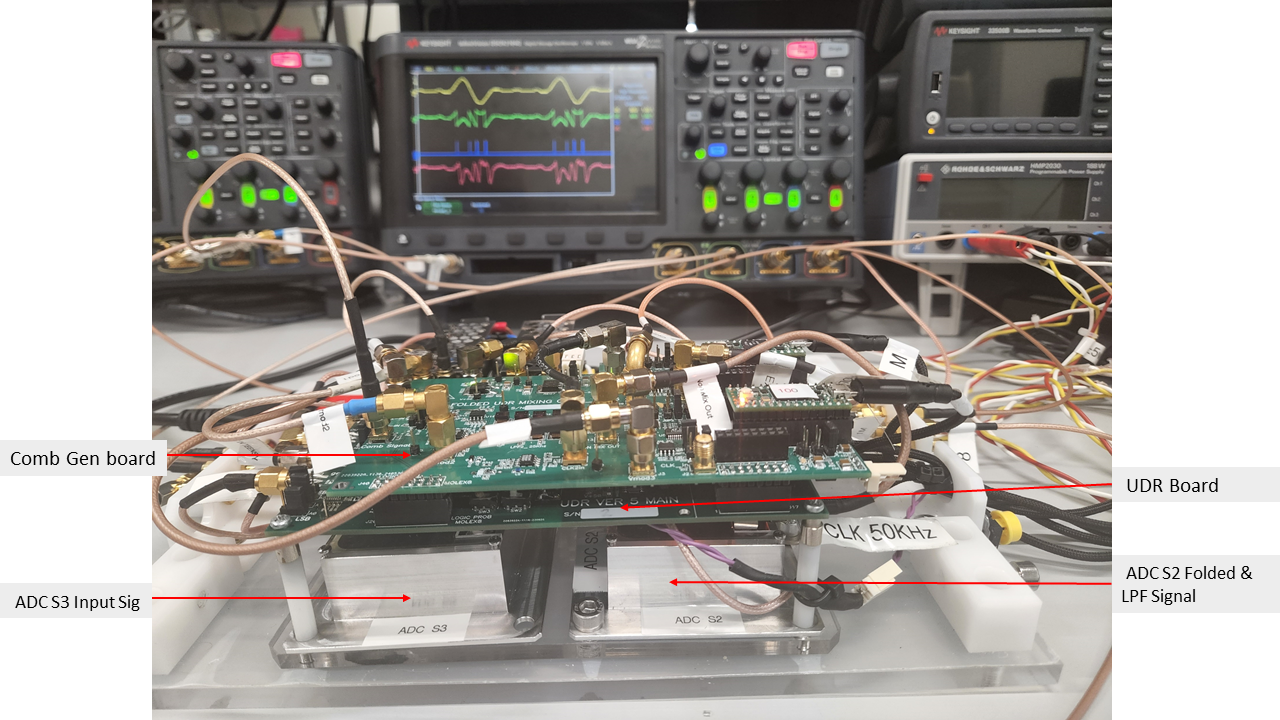}
  \includegraphics[width=0.8\textwidth]{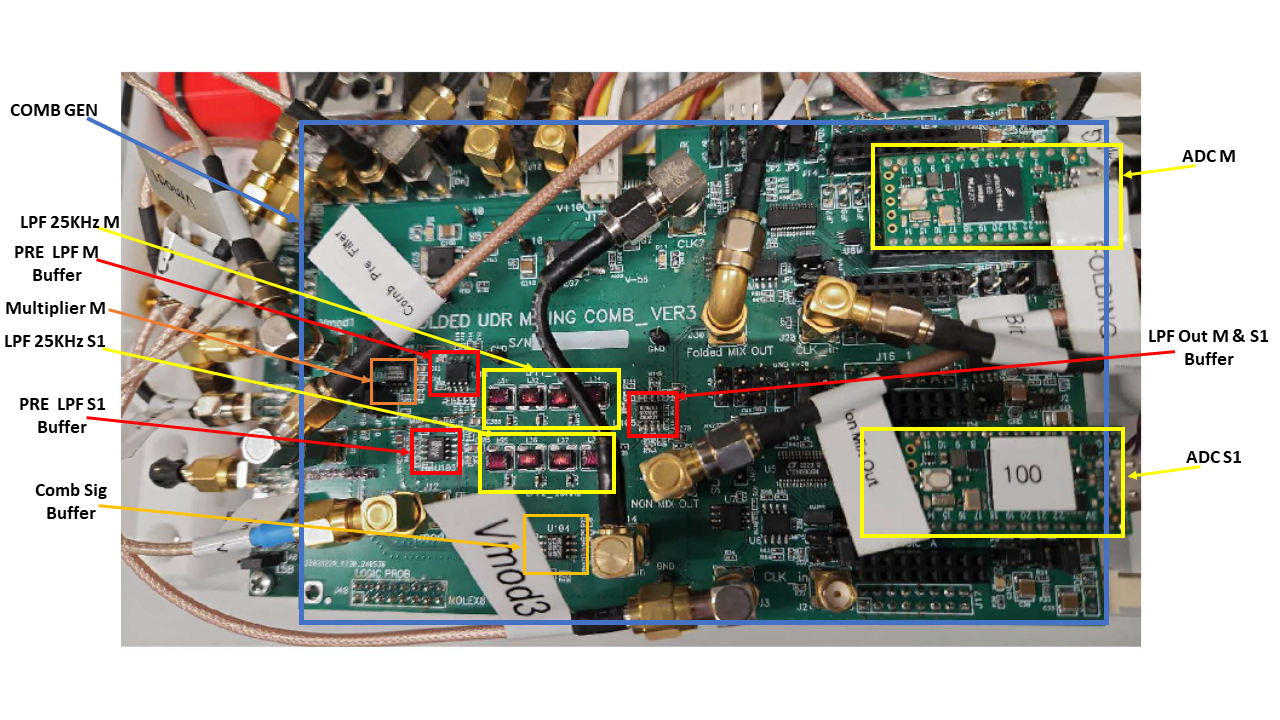}
  \caption{Realistic Modulo hardware board with highlighted components shown from two perspectives: the top image highlights key components such as the comb block described in Section \ref{sec:HW_comb}, while the bottom image highlights different components from a side view.}
  \label{fig:HW_board}
\end{figure*}

In our setting the sampling rate will be $T_s=$50KHz. Fig. \ref{Fig: System Block Diagram} depicts the setup with four different analog input channels. All four ADCs are synchronized with the sampling clock at 50KHz generated by the TEENSY-M. The ADCs are capable of sampling at 10MHz, allowing them to handle frequencies much higher than the current sampling rate. The analog modulo operation uses a board, as described in \cite{mulleti2023hardware}. This will be referred to as Unlimited Dynamic Range (UDR) in the diagram. We first present the high-level system architecture, followed by a detailed discussion of the hardware components.

\begin{figure}[ht]
  \centering
  \includegraphics[width=1\columnwidth]{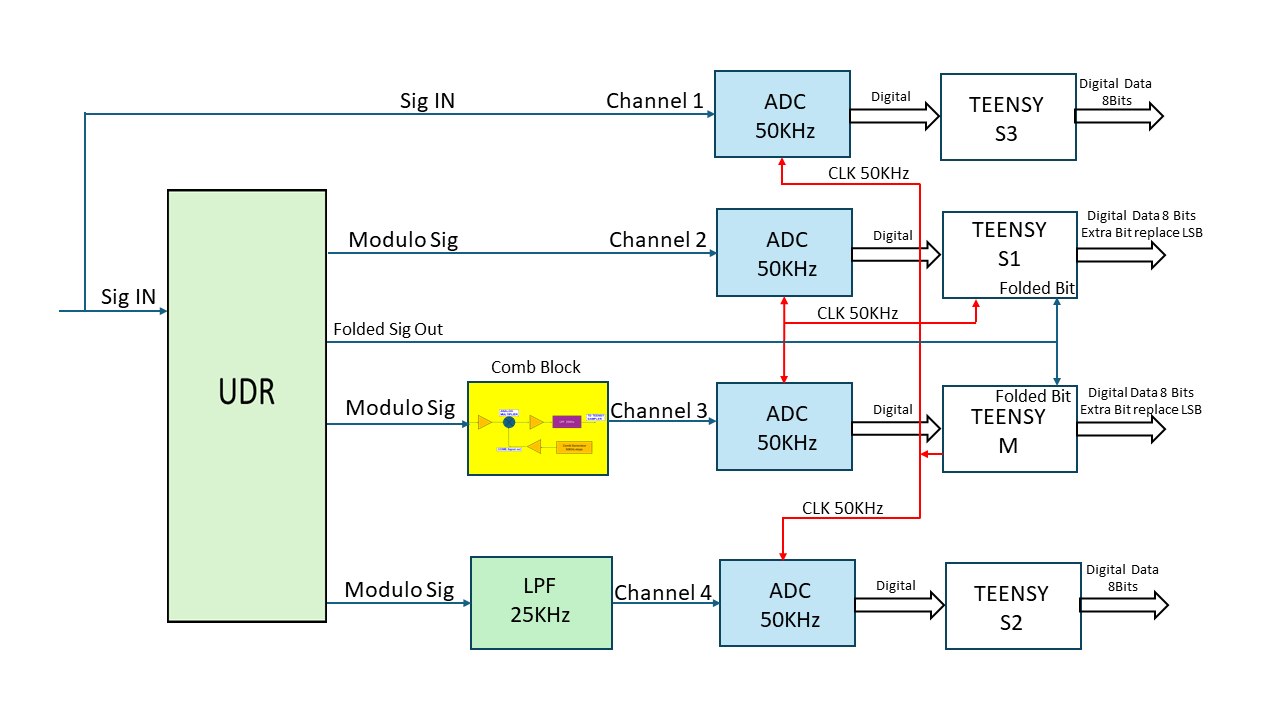}
  \caption{Four analog input signals with four synchronized 50KHz ADC samplers.}
  \label{Fig: System Block Diagram}
\end{figure}

\subsection{High-Level Design}\label{sec:high_level}

Sampler 1 (TEENSY-M) in channel 3 is the proposed sampling approach, the modulo signal multiplied by a 50KHz comb generator, followed by a 25KHz LPF as shown in Fig. \ref{fig:sys comb extrabit}.
Sampler 2 (TEENSY-S1) receives the modulo signal generated by the UDR at Channel 2, functioning as an ideal sampler. This channel implements the setup shown in Fig. \ref{fig:azar} and follows current modulo sampling systems \cite{azar2022robust, Azar_Mulleti_Eldar_2022a, mulleti2023hardware, Shah_Mulleti_Eldar_2023}.
Sampler 3 (TEENSY-S3) samples the input signal (Sia). This is traditional sampling. 
Sampler 4 (TEENSY-S2) samples the UDR modulo signal with a 25KHz LPF at Channel 4, representing a realistic sampler with a maximum sampling frequency of 50KHz. This configuration removes higher frequency components, resulting in errors that current algorithms do not account for. The 50KHz clock synchronizing all four ADCs is generated by TEENSY-M.
We implement the extra bit described in \cite{bernardo2024modulo} and discussed in Section \ref{sec:mod_rec} as part of our recovery approach; more details are provided later in this section.

\begin{figure}[ht]
  \centering
  \includegraphics[width=1\columnwidth]{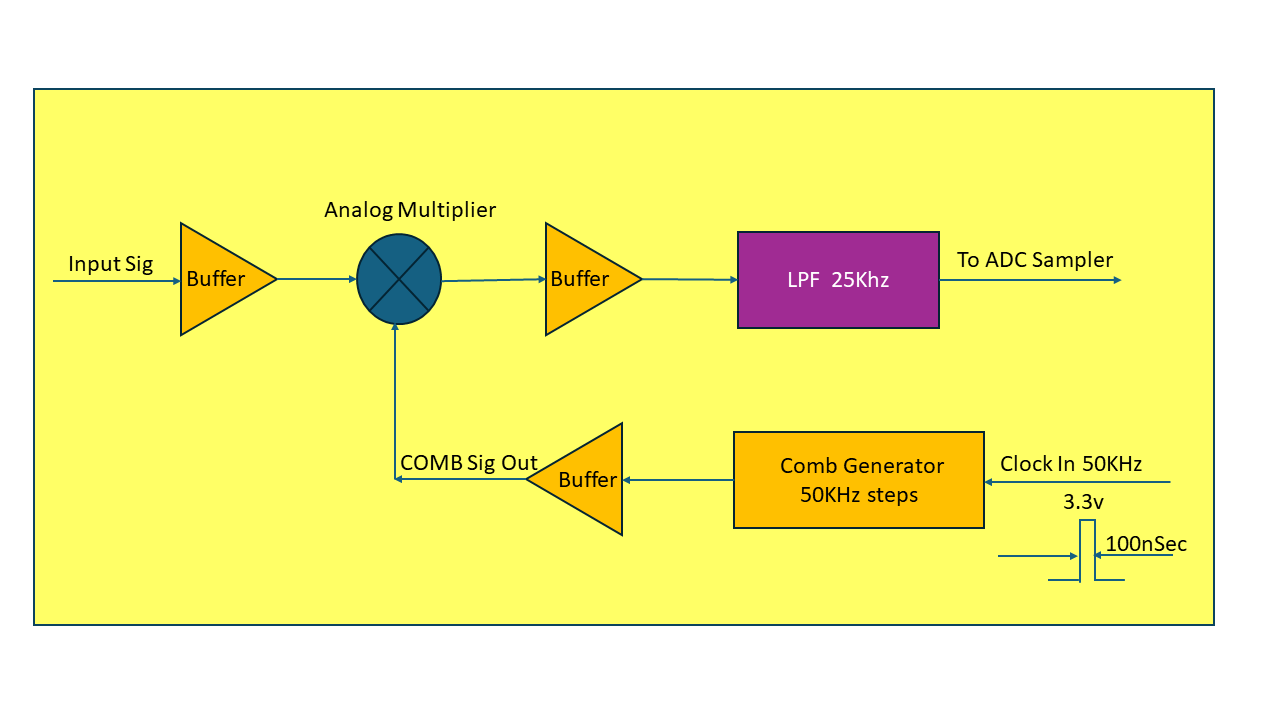}
  \caption{Multiplied signal with COMB Generator at 50KHz steps, passed through a 25KHz LPF, then to the 50KHz ADC sampler.}
  \label{Fig: COMB Gen Block Diagram}
\end{figure}


\begin{figure}[ht]
  \centering
  \includegraphics[width=1\columnwidth]{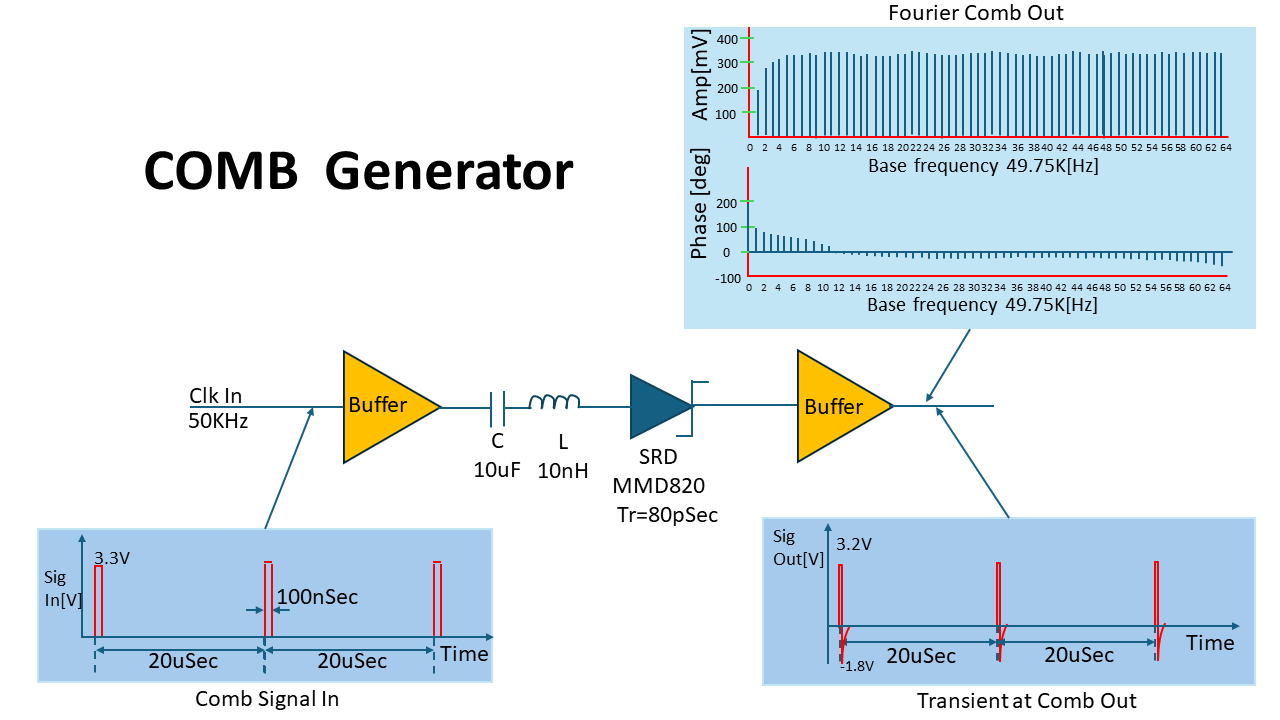}
  \caption{COMB Generator with 50KHz steps and SRD.}
  \label{Fig: Comb Gen Diagram}
\end{figure}

\subsection{Multiplier and Comb Generator Design}
\label{sec:HW_comb}

Next we present the hardware implementation of generating the delta comb and analog multiplier as outlined in Section \ref{sec:sampling} and shown in Figs. \ref{fig:sys comb} and \ref{fig:sys comb extrabit}. The design involves a Step-Recovery Diode (SRD) in Fig. \ref{Fig: Comb Gen Diagram}, based on a comb generator and multiplier circuits reactively terminated at the output port with higher-order harmonics. Power is partially reflected and recombined to produce stronger harmonics. The damping factor $\ddot{I}$ of the SRD, a primary driver in frequency-multiplier or comb-generator design, is defined as follows:

\begin{equation}
\ddot{I} = \left( \frac{1}{2R} \right) \left( \frac{L}{C} \right)^{  \frac{1}{2} }
\end{equation}
where \(L\) is the diode’s inductance, \(C\) is the diode’s reverse capacitance, and \(R\) is the load resistance. The damping factor $\ddot{I}$ is between 0.4 and 0.5. If damping is too low, stability problems can arise; if too high, the output pulse becomes too long. An SRD, also known as a snap-off diode or charge-storage diode, generates extremely short pulses. The main phenomenon used in SRDs is the storage of electric charge during forward conduction due to the finite lifetime of minority carriers in semiconductors.

The SRD Diode used in the prototype is the MACOM MMD820, with a lifetime \(\tau = 60\)ns and a transition time \(Tr = 80\)ps, enabling frequency pulses over 1GHz.

\begin{figure}[ht]
  \centering
  \includegraphics[width=1\columnwidth]{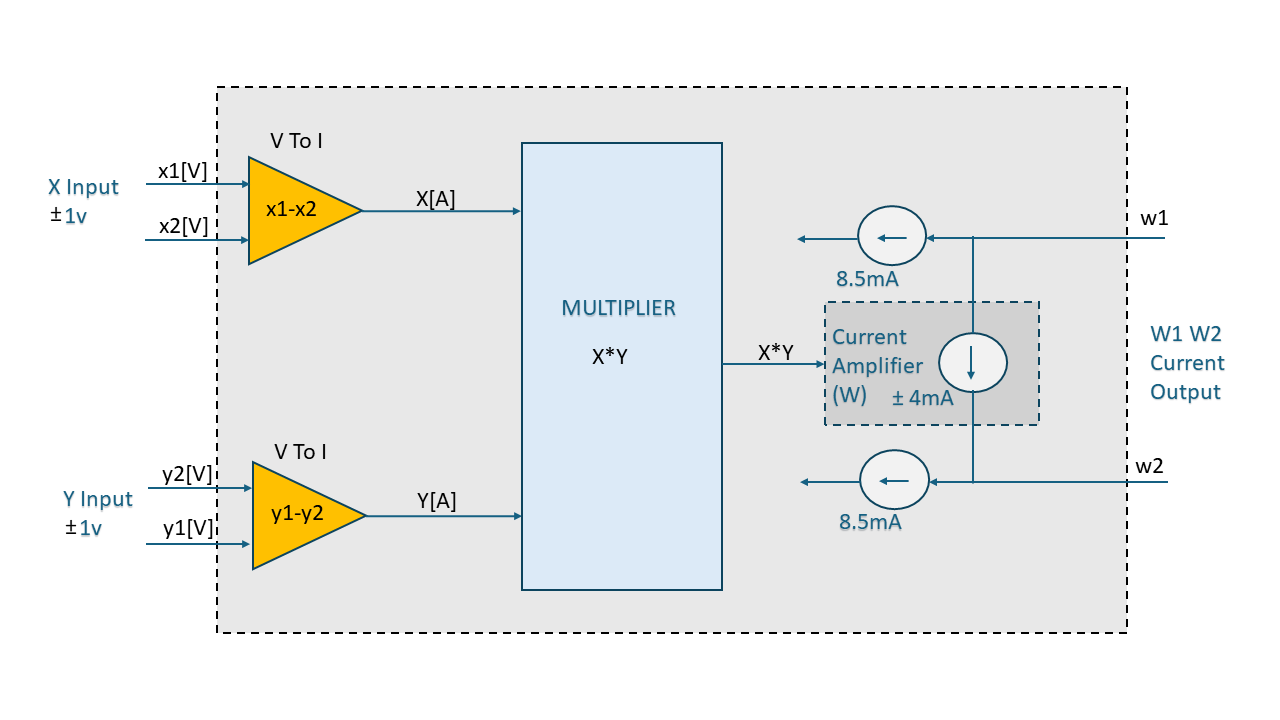}
  \caption{Functional equivalent of the analog multiplier AD834.}
  \label{Fig: AD834 Diagram}
\end{figure}

In Fig. \ref{Fig: AD834 Diagram}, the analog multiplier used is the AD834, comprised of three differential signal interfaces: two voltage inputs (X = X1 - X2, Y = Y1 - Y2) and the current output (W). The current flows in the direction shown in Fig. \ref{Fig: AD834 Diagram} when X and Y are positive. The outputs (W) have a standing current of typically 8.5mA. The input voltages are first converted to differential currents that drive the trans-linear core, with the equivalent resistance of the voltage-to-current (V-I) converters being about 285$\Omega$. The output current W is the linear product of input voltages (X and Y) divided by (1V)\(^2\) and multiplied by the scaling current of 4mA. With inputs specified in volts, the simplified expression is W = (XY)4mA. The outputs appear as a differential pair of currents at open collectors, requiring external current-to-voltage conversion for a single-ended, ground-referenced voltage output.

\subsection{Extra-bit Information}

When the input signal is folded, this event is captured using an interrupt on the TEENSY processor. The occurrence of a fold replaces the least significant bit (LSB) of the sampled word with the next clock signal to the ADC, which is generated by the processor at a frequency of 50KHz. This extra-bit indicates the folding points in time. It is important to note that the LSB only indicates that a fold has occurred, without specifying the direction or the number of folds. After the LSB information has been replaced, the processor clears the previously stored folding bit information. This process continues each time the input signal exceeds a threshold and is folded. During digital processing, the system interprets the LSB as folding information rather than part of the measurement.

\subsection{LPF Implementation and Frequency Profile}

The 25KHz LPF depicted in Fig. \ref{Fig: LPF25KHz} is a 7th order Butterworth filter. It features no ripple in both the passband and stopband, ensuring a frequency response that is as flat as possible within the passband. Refer to Fig. \ref{Fig: AttenVerFreqLPF} for the attenuation versus frequency of our LPF.

\begin{figure}[ht]
  \centering
  \includegraphics[width=1\columnwidth]{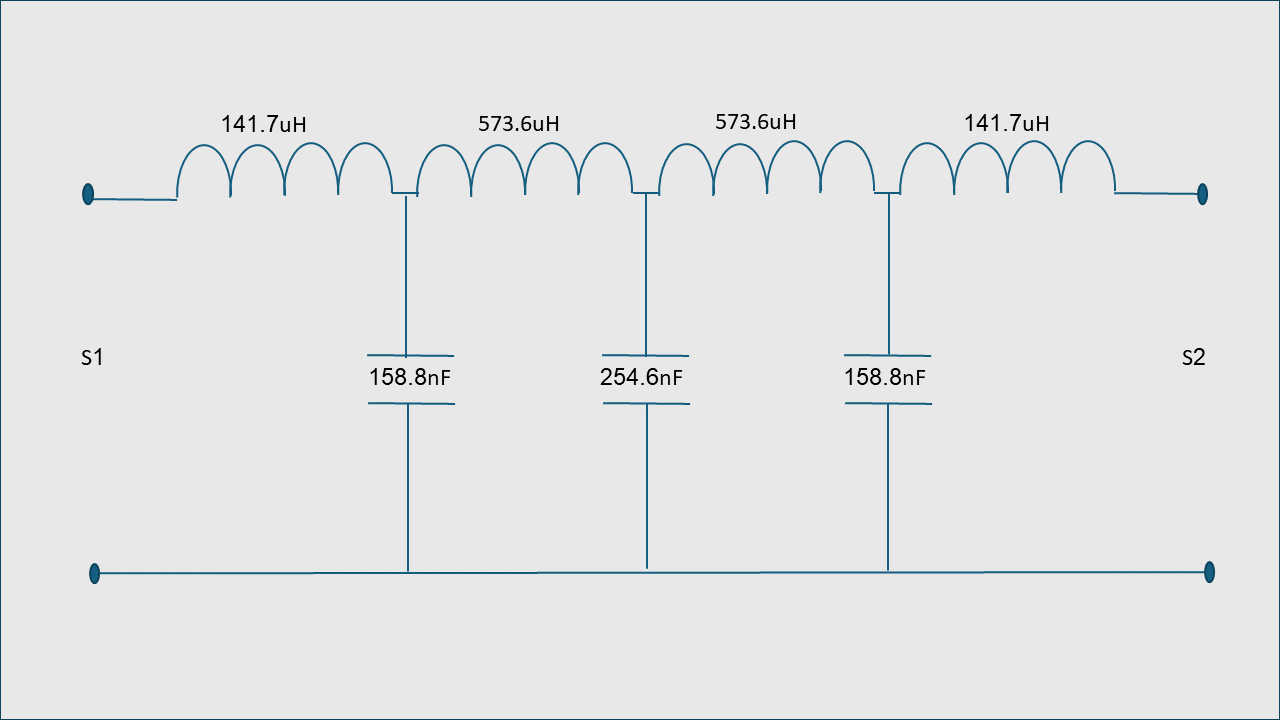}
  \caption{25KHz LPF hardware implementation. }
  \label{Fig: LPF25KHz}
\end{figure}

\begin{figure}[ht]
  \centering
  \includegraphics[width=1\columnwidth]{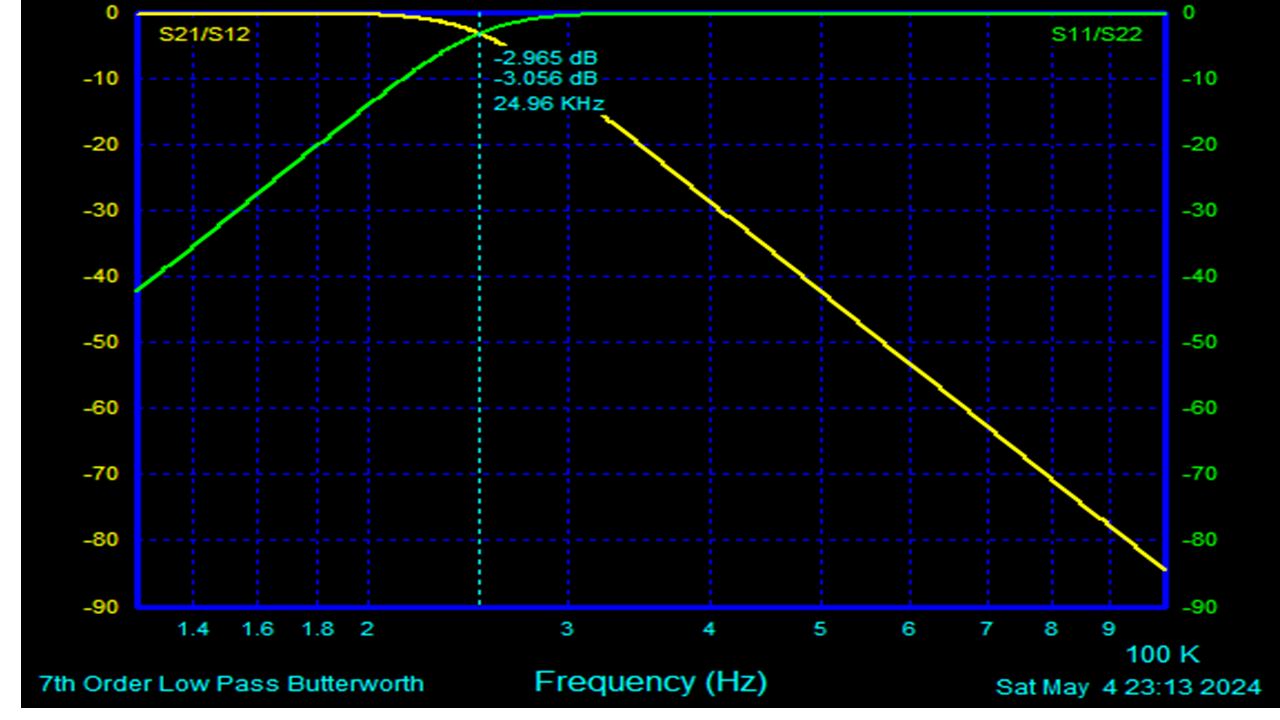}
  \caption{Attenuation versus frequency LPF 25KHz S parameters. In the measured response S21 and S12, the cut-off frequency (-3dB) is about 25KHz.}
  \label{Fig: AttenVerFreqLPF}
\end{figure}

\begin{table*}[ht]
\centering
\begin{tabular}{l|l|l}
{\bf Component} &{\bf Model number} &{\bf Maker}\\ \hline
ADC 12 bits 10Msps &LTC1420 &LINEAR TECHNOLOGY\\
SRD &MMD820 &MACOM\\
4 Quadrant Multiplier &AD834 &ANALOG DEVICES \\ 
LPF 25KHz &7th order Butterworth &Discrete components\\ 
MicroController board &Teensy4 &ARM
\end{tabular}
\caption{Hardware Components List.}
\label{tab:comp}
\end{table*}

\subsection{Hardware Implementation}

The following figures showcase the final hardware prototype and its operation. Fig. \ref{fig:HW_board} presents the hardware, providing a tangible representation of the design concepts discussed in the previous sections.
The oscilloscope snapshot in Fig. \ref{fig:oscilloscope_snapshot} illustrates the hardware in action. The yellow line represents the input BL signal with maximal frequency of 1KHz, while the green line shows the folded signal after the modulo operation. The blue line displays the folded signal multiplied by the comb, with the overlaid green and blue lines clearly demonstrating the effect of this multiplication. The red line represents the signal after passing through the LPF filter, which is the signal to be sampled as part of our proposed pipeline. As predicted by Theorem \ref{thm:mixer2}, the samples from the green and red lines are expected to align, and this alignment is visually confirmed in the snapshot. Table \ref{tab:comp} provides a list of the components used, along with their model numbers and manufacturers.


\begin{figure}[ht]
  \centering
  \includegraphics[width=1\columnwidth]{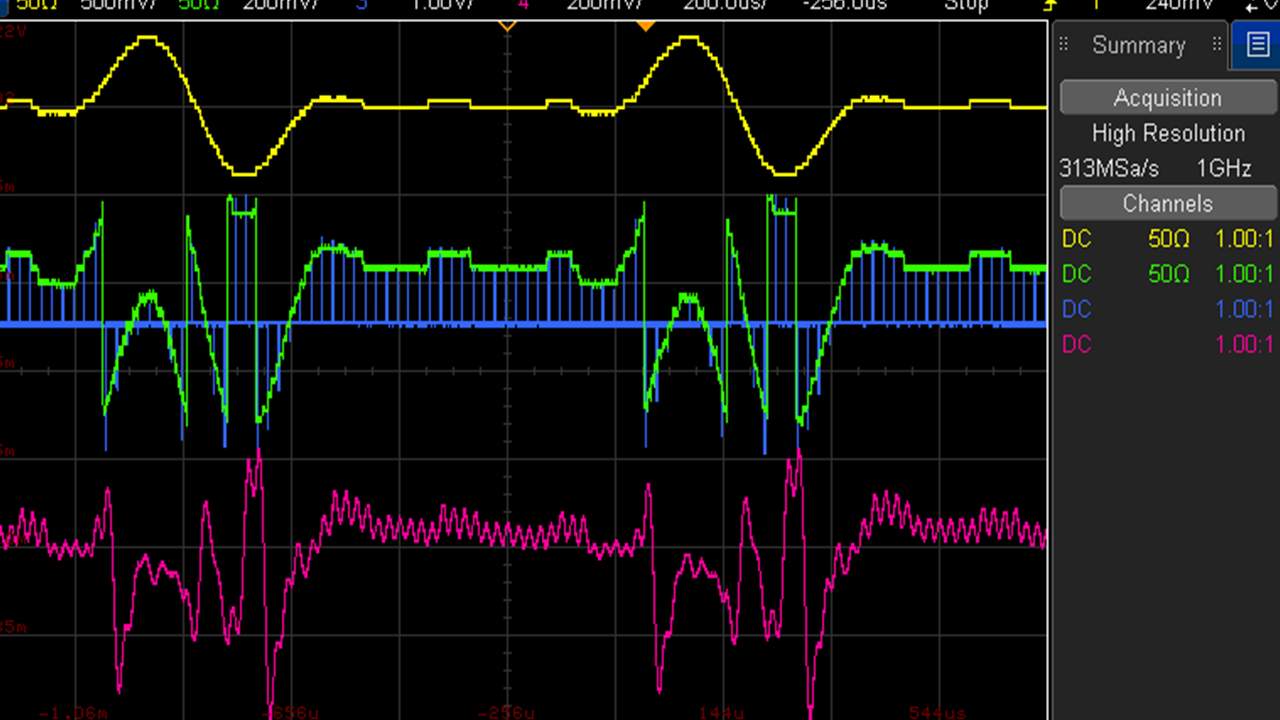}
  \caption{Oscilloscope snapshot: yellow—input BL signal, green—folded signal, blue—folded signal after comb multiplication, red—post-LPF signal for sampling.}
  \label{fig:oscilloscope_snapshot}
\end{figure}

\section{Results}
\label{sec:results}

\subsection{Simulated Results}

We generated a total of 500 random BL signals with maximal frequency of 5 KHz. The signal is of the form
\begin{equation}
	x(t) = \sum_{i=1}^{98}{a_i \text{sinc}\Big(\frac{t-i T}{T}\Big)} \quad 0\le t\le 1
\end{equation}
where $T$ is the Nyquist rate of $10^{-4}$ second, $\{a_i\}$ are i.i.d uniform in the range of $[-\frac{1}{2}, \frac{1}{2}]$. Our analysis covers scenarios with 6 and 8 bits, various oversampling rates, and the inclusion or exclusion of the comb generator.

In the recovery approach section, we specify that one bit is used for cross-level information in our modulo sampling setup. Thus, in settings where the total number of bits is defined, the modulo method allocates one bit for this purpose and the remainder for sampling, while classical sampling uses all bits directly for sampling. Despite this bit allocation, our results show that modulo sampling achieves a better MSE compared to the classical method.

Fig. \ref{fig:sim_raw} showcase the results in both 6-bit and 8-bit systems without the use of the comb generator. The term ``classical-theoretical" refers to the quantization error as calculated theoretically, as shown in equation (\ref{eq:claasical_quant}). In contrast, ``classical-live" represents the practical quantization error, which aligns closely with the theoretical predictions. ``Modulo-live" captures the error observed during practical modulo sampling, while ``modulo-ideal sampler" denotes the error when using an ideal modulo sampler, demonstrating superior performance over the classical approach which assumes infinite DR. Despite this, the overall error in the modulo method remains higher, primarily due to \(E_{\text{mod-HF}}\), which shown in the figure as ``mod-HF". This underscores the significant impact of this error and emphasizes the need for a comb generator to optimize performance.

The performance comparison with the use of a comb generator is presented in Fig. \ref{fig:sim_mixer},  highlights a substantial improvement. The labeling remains consistent with Fig. \ref{fig:sim_raw}. The inclusion of a comb generator significantly reduces \(E_{\text{mod-Q}}\), bringing it well below the quantization error. As a result, the final error in modulo recovery approaches that of an ideal sampler and is considerably lower than the error observed in the classical approach for oversampling rate of at least 5. This demonstrates the effectiveness of the comb generator in improving the performance of modulo sampling.

The error \(E_{\text{mod-HF}}\) occurs during the LPF stage and is unaffected by the quantization step, as those are independent. For systems with practical bit counts, such as 8-bit ADCs, achieving a lower \(E_{\text{mod-HF}}\) is crucial to reduce the overall reconstruction error, as shown in (\ref{eq:mod_mse}). By employing our approach, which effectively mitigates \(E_{\text{mod-HF}}\), the final reconstruction error closely matches that of the ideal modulo sampler. This enables us to surpass the performance of classical sampling methods while using realistic components.

\begin{figure*}
	\centering
	\includegraphics[width=0.4\textwidth]{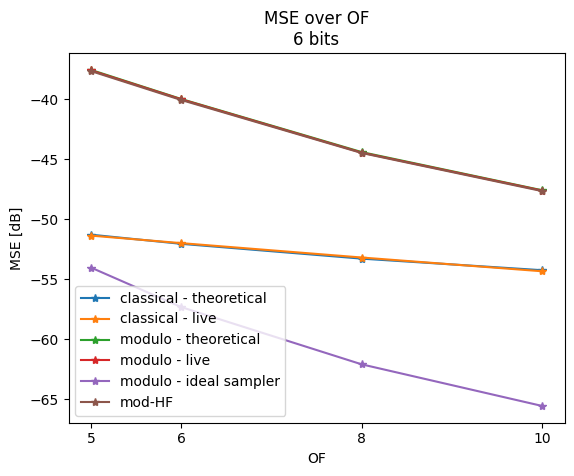}
	\includegraphics[width=0.4\textwidth]{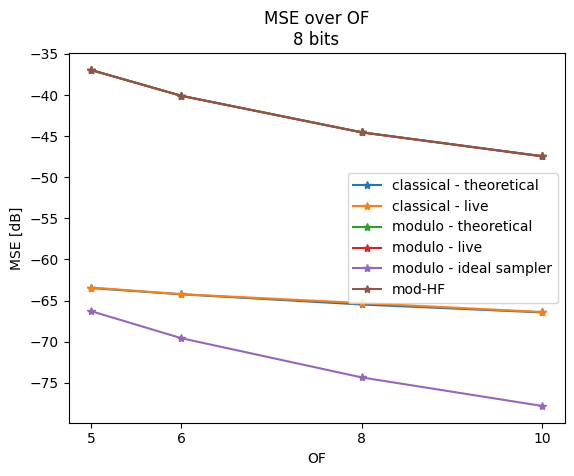}
	\caption{Comparison of quantization errors in 6-bit (left) and 8-bit (right) systems without the comb generator.}
	\label{fig:sim_raw}
\end{figure*}


\begin{figure*}
	\centering
	\includegraphics[width=0.4\textwidth]{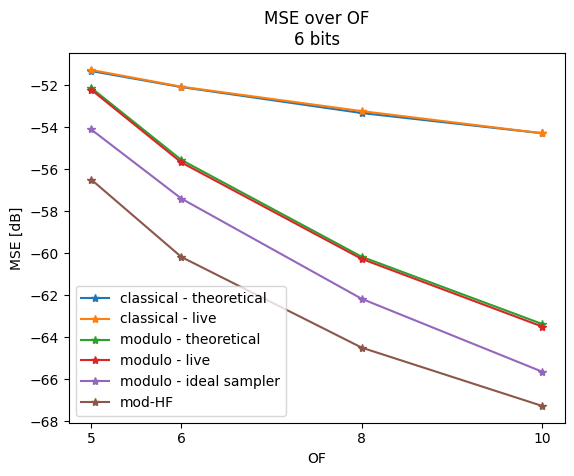}
	\includegraphics[width=0.4\textwidth]{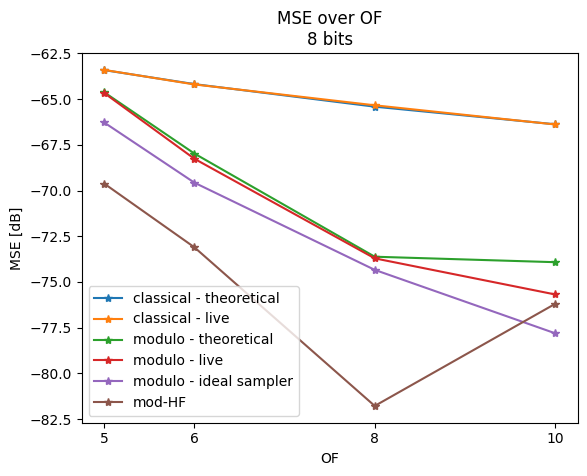}
	\caption{ Comparison of quantization errors in 6-bit (left) and 8-bit (right) systems incorporating the comb generator.
}
	\label{fig:sim_mixer}
\end{figure*}

\subsection{Hardware Results}
\label{sec:Hw_results}

We generated a test BL signal with a maximum frequency of 1KHz and passed it through the hardware setup as shown in Fig. \ref{fig:oscilloscope_snapshot}. The signal was then sampled using the four ADCs mentioned in \ref{sec:high_level}.
In Fig. \ref{fig:HW_measure}, we present the samples obtained from Samplers 1, 2, and 3. Sampler 1 represents our proposed realistic hardware pipeline, Sampler 2 captures the folded signal measured with an ideal sampler, and Sampler 3 records the folded signal after passing through a LPF, as would be the case with a non-ideal, realistic sampler. A strong similarity is observed between the measurements from Samplers 1 and 2, as predicted by Theorem \ref{thm:mixer2}. Furthermore, the samples from Sampler 3 display significantly different measurements, highlighting the necessity of the proposed hardware configuration. This discrepancy further supports the findings from our simulations, underscoring the importance of the hardware design in achieving accurate signal recovery.

We used straightforward unwrapping method \cite{tribolet1977new, goldstein1988satellite} for the data from Samplers 1 and 2.
Sampler 1, representing our proposed realistic hardware pipeline, we applied a standard unwrapping algorithm, unwrapping only when the extra-bit information indicated a fold. For Sampler 2, which measured the folded signal with an ideal sampler, we also used the standard unwrapping algorithm, but without relying on extra-bit information.
We also included results from Sampler 4, which sampled the input signal using a classical sampling approach.
The recovered signals are plotted in Fig. \ref{fig:HW_rec}, showcasing the recovery of the modulo signal using realistic hardware, an ideal sampler, and a classical sampler. Subfigure (a) displays the recovered BL signal, while subfigure (b) shows all the signals after applying a digital LPF corresponding to the known maximum frequency of the input signal. To ensure an accurate comparison, the signals have been aligned to account for a slight unintentional delay introduced by the hardware.

\begin{figure}[ht]
  \centering
  \includegraphics[width=1\columnwidth]{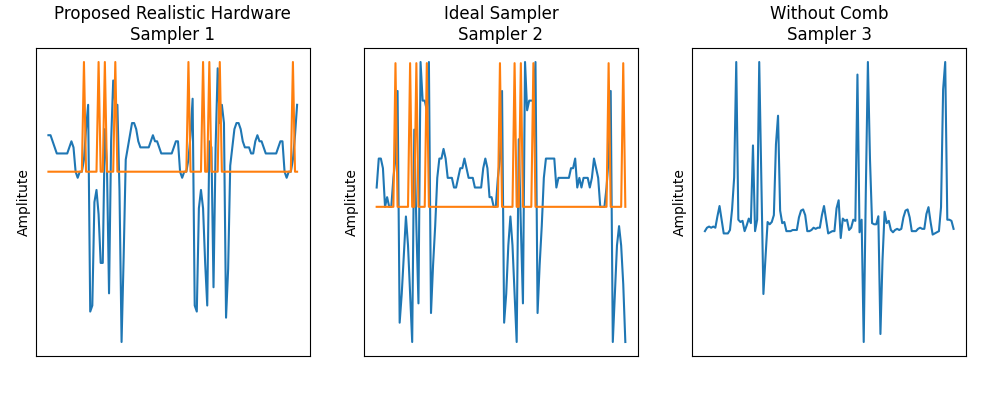}
  \caption{Comparison of sampled signals: Right shows the proposed realistic sampler (Sampler 1) with 7-bit samples (blue) and extra-bit information (orange). Middle depicts the ideal sampler (Sampler 2) with similar results. Left illustrates the realistic case (Sampler 4) with noticeably different measurements.}
  \label{fig:HW_measure}
\end{figure}

\begin{figure}[ht]
  \centering
  \includegraphics[width=1\columnwidth]{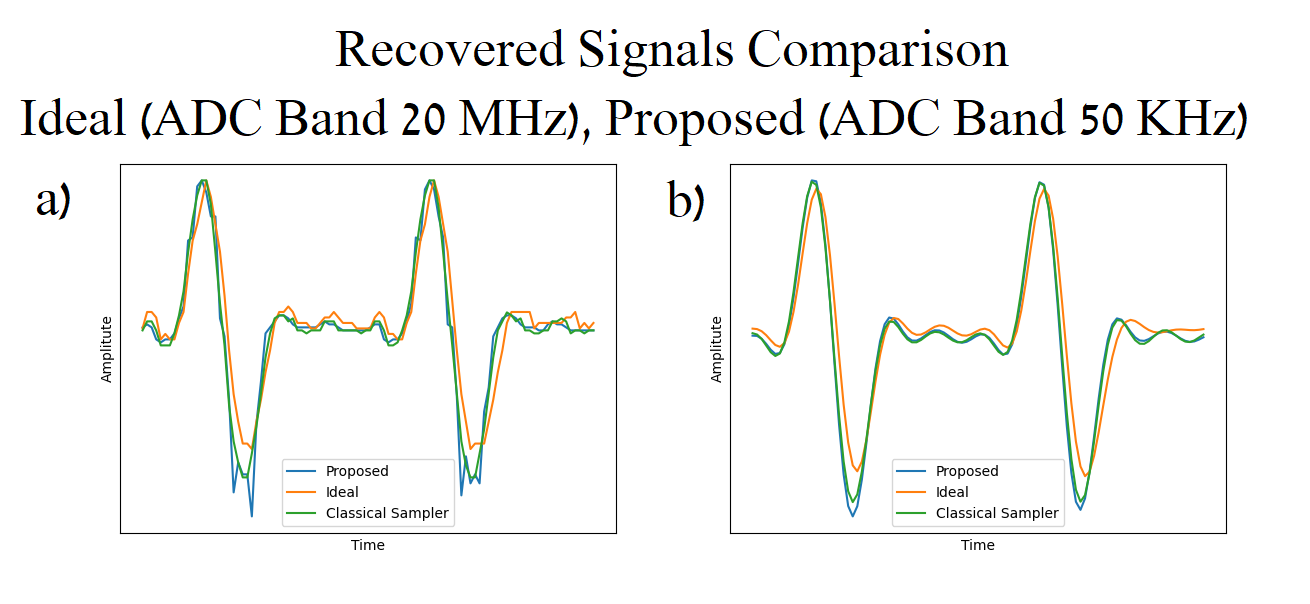}
  \caption{Comparison of recovered signals: (a) Recovered BL signal using Sampler 1 (proposed approach with realistic components), Sampler 2 (ideal sampler), and Sampler 4 (classical sampling). (b) Signals after applying a digital LPF for the input signal's maximum frequency.}
  \label{fig:HW_rec}
\end{figure}

\section{Conclusion}
\label{sec:conclusion}

In this paper, we tackled the critical issue of high-frequency components introduced by the modulo operator. Existing methods often overlook this, assuming an ideal sampler that can handle much higher frequencies, which leads to the unnecessary use of more powerful, energy-intensive ADCs. Our sampling approach employs an analog mixer and a LPF to produce a BL signal, with samples at its Nyquist rate aligning with the ideal samples of the modulo signal. This ensures that the ADC only needs to process frequencies within the intended sampling range. Following this step, any existing modulo recovery method for BL signals can be applied effectively. We implemented this approach in hardware and validated it through experiments, underscoring the necessity of addressing high-frequency components in the sampling process.
We also demonstrated that the proposed method achieves performance comparable to the ideal modulo sampler and surpasses the classical sampler without DR limitations.
Our findings show that modulo recovery can be achieved efficiently using realistic ADCs, making this method highly suitable for real-world applications.


\bibliographystyle{IEEEtran}
\bibliography{refs}

\end{document}